\documentclass[10pt,conference]{IEEEtran}

\usepackage{graphicx}
\usepackage{enumitem}
\usepackage[cmex10]{amsmath}
\usepackage{amssymb}
\usepackage{cite}
\usepackage{amsthm}
\usepackage{textcomp}
\usepackage{color}
\usepackage{cases}
\usepackage{tcolorbox}
\usepackage{fix-cm}
\usepackage{bm}
\usepackage{multirow}
\usepackage{tikz}
\usepackage{pgfplots}
\usepackage{array}
\newcolumntype{P}[1]{>{\centering\arraybackslash}p{#1}}
\newcolumntype{M}[1]{>{\centering\arraybackslash}m{#1}}

\usetikzlibrary{pgfplots.groupplots}

\newcommand{\vct}[1]{\bm{#1}}

\makeatletter
\def\thm@space@setup{\thm@preskip=2pt
\thm@postskip=2pt \itshape}
\makeatother
\newtheoremstyle{newstyle}      
{} 
{} 
{\mdseries} 
{} 
{\bfseries} 
{.} 
{ } 
{} 

\theoremstyle{newstyle}

\newtheorem{theorem}{Theorem}
\newtheorem{lemma}{Lemma}

\theoremstyle{definition}

\newtheorem{definition}{Definition}

\theoremstyle{remark}
\newtheorem{remark}{Remark}

\setlist[description]{style=multiline}

\IEEEoverridecommandlockouts

\begin{document}
\sloppy

\setlength{\abovedisplayskip}{1mm}
\setlength{\belowdisplayskip}{1mm}
\setlength{\abovecaptionskip}{1mm}
\setlength{\belowcaptionskip}{-6pt}

\title{Near-Optimal Straggler Mitigation for Distributed Gradient Methods} 


\author{\IEEEauthorblockN{Songze~Li$^*$, Seyed~Mohammadreza~Mousavi~Kalan$^*$, A.~Salman~Avestimehr, and Mahdi~Soltanolkotabi} \thanks{$^*$The first two authors contributed equally to this work.}
\IEEEauthorblockA{University of Southern California \\Email: \{songzeli,mmousavi\}@usc.edu, avestimehr@ee.usc.edu, soltanol@usc.edu}
}

\maketitle

\begin{abstract} 
Modern learning algorithms use gradient descent updates to train inferential models that best explain data. Scaling these approaches to massive data sizes requires proper distributed gradient descent schemes where distributed worker nodes compute partial gradients based on their partial and local data sets, and send the results to a master node where all the computations are aggregated into a full gradient and the learning model is updated. However, a major performance bottleneck that arises is that some of the worker nodes may run slow. These nodes a.k.a.~stragglers can significantly slow down computation as the slowest node may dictate the overall computational time. We propose a distributed computing scheme, called Batched Coupon's Collector (BCC) to alleviate the effect of stragglers in gradient methods. We prove that our BCC scheme is robust to a near optimal number of random stragglers. We also empirically demonstrate that our proposed BCC scheme reduces the run-time by up to 85.4\% over Amazon EC2 clusters when compared with other straggler mitigation strategies. We also generalize the proposed BCC scheme to minimize the completion time when implementing gradient descent-based algorithms over heterogeneous worker nodes.

\end{abstract}

\section{Introduction}\label{sec:intro}
Gradient descent (GD) serves as a working-horse for modern inferential learning tasks spanning computer vision to recommendation engines. In these learning tasks one is interested in fitting models to a training data set of $m$ training examples $\{\vct{x}_j\}_{j=1}^m$ (usually consisting of input-output pairs). The fitting problem often consists of finding a mapping that minimizes the empirical risk
\begin{align*}
\mathcal{L}(\vct{w}):=\frac{1}{m}\sum_{j=1}^m \ell(\vct{x}_j;\vct{w}).
\end{align*}
Here, $\ell(\vct{x}_j;\vct{w})$ is a loss function measuring the misfit between the model and output on $\vct{x}_j$ with $\vct{w}$ denoting the model parameters.
GD solves the above optimization problem via the following iterative updates
\begin{align}
        \vct{w}_{t+1}=\vct{w}_t-\mu_t\nabla \mathcal{L}(\vct{w}_t) = \vct{w}_t - \mu_t \frac{1}{m} \sum_{j=1}^m \vct{g}_j(\vct{w}_t).\label{eq:GD-update}
\end{align}
Here, $\vct{g}_j(\vct{w}_t) = \nabla \ell (\vct{x}_j;\vct{w}_t)$ is the partial gradient with respect to $\vct{w}_t$ computed from $\vct{x}_j$, and $\mu_t$ is the learning rate in the $t$th iteration. 

In order to scale GD to handle massive amount of training data, developing parallel/distributed implementations of gradient descent over multiple cores or GPUs on a single machine, or multiple machines in computing clusters is of significant importance \cite{recht2011hogwild,gemulla2011large,zhuang2013fast,seide20141}.
\begin{figure}[htbp]
  \centering
  \includegraphics[width=0.45\textwidth]{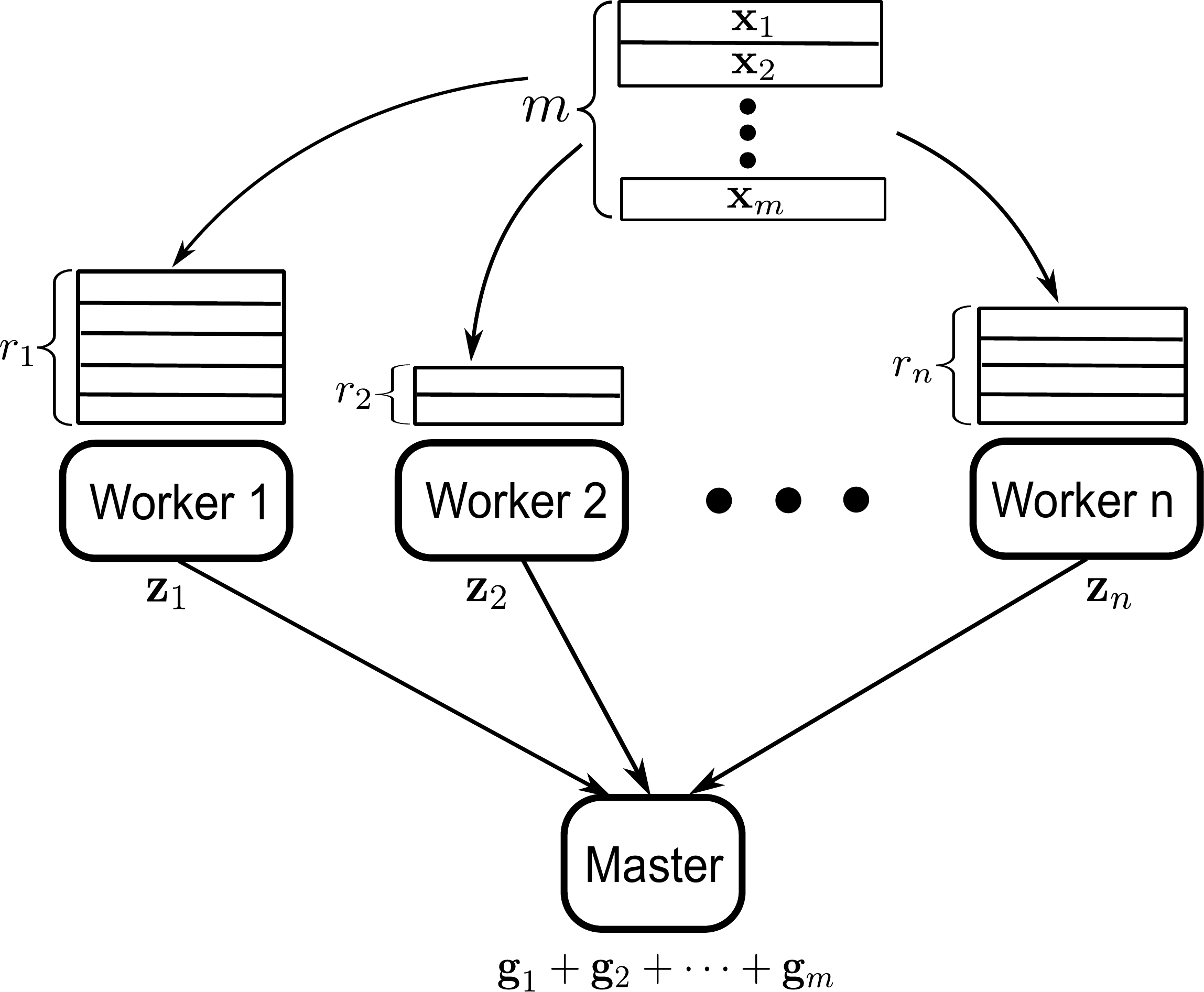}
  \caption{A master-worker distributed computing model for distributed gradient descent.}
  \label{fig:setting}
\end{figure}
In this paper we consider a distributed computing model consisting of a master node and $n$ workers as depicted in Fig.~\ref{fig:setting}. Each worker~$i$ stores and processes a subset of $r_i$ training examples locally, and then generates a message $\vct{z}_i$ based on computing partial gradients using the local training data, and then sends this message to the master node. The master collects the messages from the workers, and uses these messages to compute the total gradient and update the model via (\ref{eq:GD-update}). If each worker processes a disjoint subset of the examples, the master needs to gather all partial gradients from all the workers. Therefore, when different workers compute and communicate at different speeds, the run-time of each iteration of distributed GD is limited by the slowest worker (or straggler). This phenomenon known as the straggler effect, significantly delays the execution of distributed computing tasks when some workers compute or communicate much slower than others. For example, it was shown in~\cite{ananthanarayanan2010reining} that over a wide range of production jobs, stragglers can prolong the completion time by 34\% at median.

We focus on straggler mitigation in the above distributed GD framework. To formulate the problem, we first define two key performance metrics that respectively characterize how much local processing is needed at each worker, and how many workers the master needs to wait for before it can compute the gradient. In particular, we define the \emph{computational load}, denoted by $r$, as the number of training examples each worker processes locally, and the \emph{recovery threshold}, denoted by $K$, as the average number of workers from whom the master collects the results before it can recover the gradient. As a function of the computational load $r$, the recovery threshold $K$ decreases as $r$ increases. For example, when $r = \frac{m}{n}$ such that each worker processes a disjoint subset of the examples, $K$ attains its maximum of $n$. One the other hand, if each worker processes all examples, i.e., $r=m$, the master only needs to wait for one of them to return the result, achieving the minimum $K=1$. For an arbitrary computational load $\frac{m}{n} \leq r \leq m$, we aim to characterize the minimum recovery threshold across all computing schemes, denoted by $K^*(r)$, which provides the maximum robustness to the straggler effect. Moreover, due to the high communication overhead to transfer the results to the master (especially for a high-dimensional model vector $\vct{w}$), we are also interested in characterizing the minimum \emph{communication load}, denoted by $L^*(r)$, which is defined as the (normalized) size of the messages received at the master before it can recover the gradient. 

To reduce the effect of stragglers in this paper we propose a distributed computing scheme, named ``Batched Coupon's Collector'' (BCC). We will show that this scheme achieves the recovery threshold
\begin{align}
     K_{\textup{BCC}}(r) = \lceil \tfrac{m}{r} \rceil H_{\lceil \frac{m}{r} \rceil}\approx \tfrac{m}{r}\log \tfrac{m}{r},
\end{align}
where $H_n$ denotes the $n$th harmonic number. We also prove a simple lower bound on the minimum recovery threshold demonstrating that
\begin{align*}
K^*(r) \geq \frac{m}{r}.
\end{align*}
Thus, our proposed BCC scheme achieves the minimum recovery threshold to within a logarithmic factor, that is,
\begin{align}
      K^*(r) \!\leq\!  K_{\textup{BCC}}(r) \!\leq\!  \lceil K^*(r)\rceil H_{\lceil \frac{m}{r} \rceil} \approx K^*(r)\log\tfrac{m}{r}.
\end{align}
We will also demonstrate that the BCC scheme achieves the minimum communication load to within a logarithmic factor, that is,
\begin{align}
      L^*(r) \leq  L_{\textup{BCC}}(r) \leq  \lceil L^*(r)\rceil H_{\lceil \frac{m}{r} \rceil} \approx L^*(r)\log\tfrac{m}{r}.
\end{align}

The basic idea of the proposed BCC scheme is to obtain the ``coverage'' of the computed partial gradients at the master. Specifically, we first partition the entire training dataset into $\frac{m}{r}$ batches of size $r$, and then each worker independently and randomly selects a batch to process. As a result, the process of collecting messages at the master emulates the coupon collecting process in the well-known coupon collector's problem (see, e.g.,~\cite{ross2012first}), which requires to collect a total of $\frac{m}{r}$ different types of coupons using $n$ independent trials.
Since the examples in different batches are disjoint, we can compress the computed partial gradients at each worker by simply summing them up, and send the summation to the master. Utilizing the algebraic property of the overall computation, the proposed BCC scheme attains the minimum communication load from each worker. 

Beyond the theoretical analysis, we also implement the proposed BCC scheme on Amazon EC2 clusters, and empirically demonstrate performance gain over the state-of-the-art straggler mitigation schemes. In particular, we run a baseline uncoded scheme where the training examples are uniformly distributed across the workers without any redundant data placement, the cyclic repetition scheme in~\cite{TLDK16} designed to combat the stragglers for the worst-case scenario, and the proposed BCC scheme, on clusters consisting of $50$ and $100$ worker nodes respectively. We observe that the BCC scheme speeds up the job execution by up to 85.4\% compared with the uncoded scheme, and by up to 69.9\% compared with the cyclic repetition scheme.  

Finally, we generalize the BCC scheme to accelerate distributed GD in heterogeneous clusters, in which each worker may be assigned different number of training examples according to its computation and communication capabilities. In particular, we derive analytically lower and upper bounds on the minimum job execution time, by developing and analyzing a generalized BCC scheme for heterogeneous clusters. 
We have also numerically evaluated the performance of the proposed generalized BCC scheme. In particular, compared with a baseline strategy where the dataset is distributed without repetition, and the number of examples a worker processes is proportional to its processing speed, we numerically demonstrate a $29.28$\% reduction in average computation time. 

\subsection*{Prior Art and Comparisons}
For the aforementioned distributed GD problem, a simple data placement strategy is that each worker selects $r$ out of the $m$ examples uniformly at random. Under this data placement, each worker processes each of the selected examples, and communicates the computed partial gradient individually to the master. Following the arguments of the coupon's collector problem, this simple randomized computing scheme achieves a recovery threshold \begin{align}
K_{\textup{random}} \approx \frac{m}{r}\log m.
\end{align}
Similar to the proposed BCC scheme, this randomized scheme achieves the minimum recovery threshold to within a logarithmic factor. However, since each worker communicates $r$ times more messages, the communication load has increased to
\begin{align}
L_{\textup{random}} \approx m\log m.
\end{align}

Recently a few interesting papers~\cite{TLDK16,halbawi2017improving,raviv2017gradient} utilize coding theory to mitigate the effect of stragglers in distributed GD. In particular, a cyclic repetition (CR) scheme was proposed in~\cite{TLDK16} to randomly generate a coding matrix, which specifies the data placement and how to encode the computed partial gradients across workers for communication. Furthermore, in~\cite{halbawi2017improving} and~\cite{raviv2017gradient}, the same performance was achieved using deterministic constructions of Reed-Solomon (RS) codes and cyclic MDS (CM) codes. These coding schemes can tolerate $r-1$ stragglers in the worst case when the computational load is $r$. More specifically, when the number of examples is equal to the number of workers ($m=n$)\footnote{When $m>n$, we can partition the dataset into $n$ groups, and view each group of $\frac{m}{n}$ training examples as a ``super example''.}, the above coded schemes achieve the recovery threshold 
\begin{align}
    K_{\textup{CR}} = K_{\textup{RS}} = K_{\textup{CM}} = m - r + 1.\label{eq:worst-thresh}
\end{align}
In all of these coded schemes, each worker encodes the computed partial gradients by generating a linear combination of them, and communicates the single coded message to the master. This yields a communication load 
\begin{align}
    L_{\textup{CR}} = L_{\textup{RS}} = L_{\textup{CM}} = m - r + 1.\label{eq:worst-comm}
\end{align}

\begin{figure}[h]
\centering
\begin{tikzpicture}[scale=1] 
\begin{axis}[
        xlabel=Computational Load ($r$),
        ylabel=Recovery Threshold ($K$), legend style={font=\small,at={(0.65,0.95)},anchor=north,legend cell align=left}]
     \addplot [black,line width=2pt] table[x index=0,y index=1]{./results1};
     \addlegendentry{Lower bound}
     
     \addplot [red,line width=2pt] table[x index=0,y index=1]{./results2};
     \addlegendentry{Proposed BCC scheme}
     
     \addplot [teal,line width=2pt] table[x index=0,y index=1]{./results3};
     \addlegendentry{Simple randomized scheme}
     
     \addplot [blue,line width=2pt] table[x index=0,y index=1]{./results4};
     \addlegendentry{CR scheme}
\end{axis}
\end{tikzpicture}
\caption{The tradeoffs between the computational load and the recovery threshold, for distributed GD using $m=100$ training examples across $n=100$ workers.}
\label{fig:compare}
\end{figure}
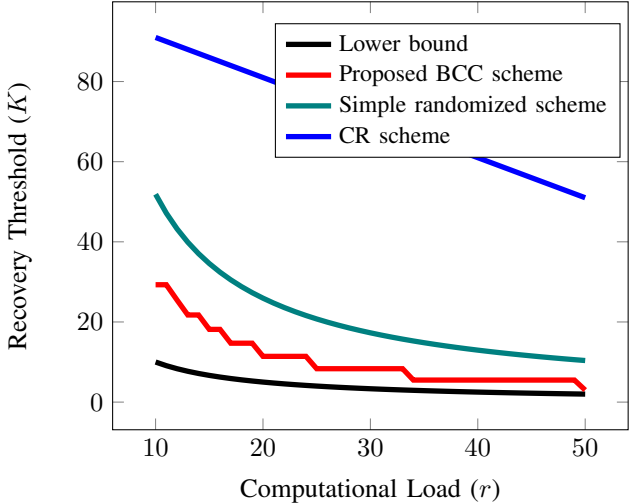

While the above simple randomized scheme and the coding theory-inspired schemes are effective in reducing the recovery threshold and the communication load respectively, the proposed BCC scheme achieves the best of both. In Fig.~\ref{fig:compare}, we numerically compare the recovery threshold of the randomized scheme, the CR scheme in~\cite{TLDK16}, and the proposed BCC scheme, and demonstrate the performance gain of BCC.  To summarize, the proposed BCC schemes has the following advantages
\begin{itemize}
   \item \emph{Simplicity}: Unlike the computing schemes that rely on delicate code designs for data placement and communication, the BBC scheme is rather simple to implement, and has little coding overhead. 

    \item \emph{Reliability}: The BCC scheme simultaneously achieves near minimal recovery threshold and communication load, enabling good straggler mitigation and fast job execution.
    
    \item \emph{Universality}: In contrast to the coding theory-inspired schemes like CR, the proposed BCC scheme does not require any prior knowledge about the number of stragglers in the cluster, which may not be available or vary across the iterations.
    
    \item \emph{Scalability}: The data placement in the BCC scheme is performed in a completely decentralized manner. This allows the BCC scheme to seamlessly scale up to larger clusters with minimum overhead for reshuffling the data. 
\end{itemize}


Finally, we highlight some recent developments of utilizing coding theory to speedup a broad calss of distributed computing tasks. In~\cite{lee2017speeding,dutta2016short}, maximum distance separable (MDS) error-correcting codes were applied to speedup distributed linear algebra operations (e.g., matrix multiplications). In particular, MDS codes were utilized to generate redundant coded computing tasks, providing robustness to missing results from slow workers. The proposed coded computing scheme in~\cite{lee2017speeding} was further generalized in~\cite{LMA16_unify}, where it was shown that the solution of~\cite{lee2017speeding} is a single operating point on a more general tradeoff between computation latency and communication load. Other than dealing with stragglers, coding theory was also shown to be an effective tool to alleviate communication bottlenecks in distributed computing. In~\cite{LMA_all,li2016fundamental}, for a general MapReduce framework implemented on a distributed computing cluster, an optimal tradeoff between the local computation on individual workers and the communication between workers was characterized, exploiting coded multicasting opportunities created by carefully designing redundant computations across workers.


\section{Problem Formulation}
%
We focus on a data-distributed implementation of the gradient descent updates in \eqref{eq:GD-update}. In particular, as shown in Fig.~\ref{fig:setting} of Section~\ref{sec:intro}, we employ a distributed computing system that consists of a master node, and $n$ worker nodes (denoted by Worker~$1$, Worker~$2,\ldots,$ Worker~$n$). Worker $i$, stores and processes locally a subset of $r_i \leq m$ training examples. We use ${\cal G}_i \subseteq \{1,\ldots,m\}$ to denote the set of the indices of the examples processed by Worker $i$. In the $t$th iteration, Worker $i$ computes a partial gradient $\vct{g}_j(\vct{w}_t)$ with respect to the current weight vector $\vct{w}_t$, for each $j \in {\cal G}_i$. Ideally we would like the workers to process as few examples as possible. This leads us to the following definition for characterizing the computational load of distributed GD schemes.
\begin{definition}[Computational Load]
We define the computational load, denoted by $r$, as the maximum number of training examples processed by a single worker across the cluster, i.e., $r := \underset{i=1,\ldots,n}{\max} r_i$.
\end{definition}




The assignment of the training examples to the workers, or the data distribution, can be represented by a bipartite graph ${\bf G}$ that contains a set of data vertices $\{d_1,d_2,\ldots,d_m\}$, and a set of worker vertices $\{k_1,k_2,\ldots,k_n\}$. There is an edge connecting $d_j$ and $k_i$ if Worker~$i$ computes $\vct{g}_j$ locally, or in other words, $j$ belongs to ${\cal G}_i$. Since each data point needs to be processed by some worker, we require that ${\cal N}(k_1)\cup \ldots \cup {\cal N}(k_n)=\{d_1,\ldots,d_m\}$, where ${\cal N}(k_i)$ denotes the neighboring set of $k_i$.
After Worker~$i$, $i=1,\ldots,n$, finishes its local computations, it communicates a function of the local computation results to the master node. More specifically, as shown in Fig.~\ref{fig:setting} Worker~$i$ communicates to the master a message $\vct{z}_i$ of the form 
\begin{align}
\vct{z}_i = \phi_i(\{\vct{g}_j: j \in {\cal G}_i\}),
\end{align}
via an encoding function $\phi_i$.

Let ${\cal W} \subseteq \{1,\ldots,n\}$ denote the index of the subset of workers whose messages are received at the master. After receiving these messages, the master node calculates the complete gradient (based on all training data) by using a decoding function $\psi$. More specifically,
\begin{align}
\psi(\{\vct{z}_i: i \in {\cal W}\}) = \frac{1}{m} \sum_{j=1}^m \vct{g}_j(\vct{w}_t).
\end{align}
In order for the master to be able to calculate the complete gradient from the received messages it needs to wait for a sufficient number of workers. We quantify this and a related parameter more precisely below. 
\begin{definition}[Recovery Threshold]
The recovery threshold, denoted by $K$, is the average number of workers from whom the master waits to collect messages before recovering the final gradient, i.e., $K := \mathbb{E}[|{\cal W}|$].
\end{definition}
\begin{definition}[Communication Load]
We define the communication load, denoted by $L$, as the average aggregated size of the messages the master receives from the workers with indices in ${\cal W}$, normalized by the size of a partial gradient computed from a single example.   
\end{definition}

We say that a pair $(r,K)$ is \emph{achievable} if for a computational load $r$, there exists a distributed computing scheme, such that the master recovers the gradient after receiving messages from on average $K$ or less workers.

\begin{definition}[Minimum Recovery Threshold]
We define the minimum recovery threshold, denoted by $K^*(r)$, as
\begin{align}
    K^*(r) := \min \{K: (r,K) \textup{ is achievable}\}
\end{align}
\end{definition}

We also define the minimum communication load, denoted by $L^*(r)$, in a similar manner.

In the next section, we propose and analyze a computing scheme for distributed GD over a homogeneous cluster, and show that it simultaneously achieves a near optimal recovery threshold and communication load (up to a logarithmic factor).

\section{The Batched Coupon's Collector (BCC) Scheme}
In this section, we consider homogeneous workers with identical computation and communication capabilities. 
As a result, each worker processes the same number of training examples, and we have $r_1 = r_2 = \cdots = r_n = r$. We note that in this case for the entire dataset to be stored and processed across the cluster, we must have $\frac{m}{r} \leq n$. For this setting, we propose the following scheme which we shall refer to as ``batched coupon's collector'' (BCC). 
\subsection{Description of BCC}
The key idea of the proposed BCC scheme is to obtain the ``coverage'' of the computed partial gradients at the master. As indicated by the name of the scheme, BCC is composed of two steps: ``batching'' and ``coupon collecting''. In the first step, the training examples are partitioned into batches, which are selected randomly by the workers for local processing. In the second step, the processing results from the data batches are collected at the master, emulating the process of the well-known coupon's collector problem. Next, we describe in detail the proposed BCC scheme.

\noindent {\bf Data Distribution.} For a given computational load $r$, as illustrated in Fig.~\ref{fig:BCC}, we first evenly partition the entire data set into $\lceil\frac{m}{r}\rceil$ data batches, and denote the index sets of the examples in these batches by ${\cal B}_1,{\cal B}_2,\ldots,{\cal B}_{\lceil\frac{m}{r}\rceil}$. Each of the batches contains $r$ examples (with the last batch possibly being zero-padded). Each worker node independently picks one of the data batches uniformly at random for local processing. We denote index set of the data points selected by Worker~$i$ as ${\cal B}_{\sigma_i}$, i.e.~${\cal G}_i = {\cal B}_{\sigma_i}$.

\begin{figure}[htbp]
  \centering
  \includegraphics[width=0.48\textwidth]{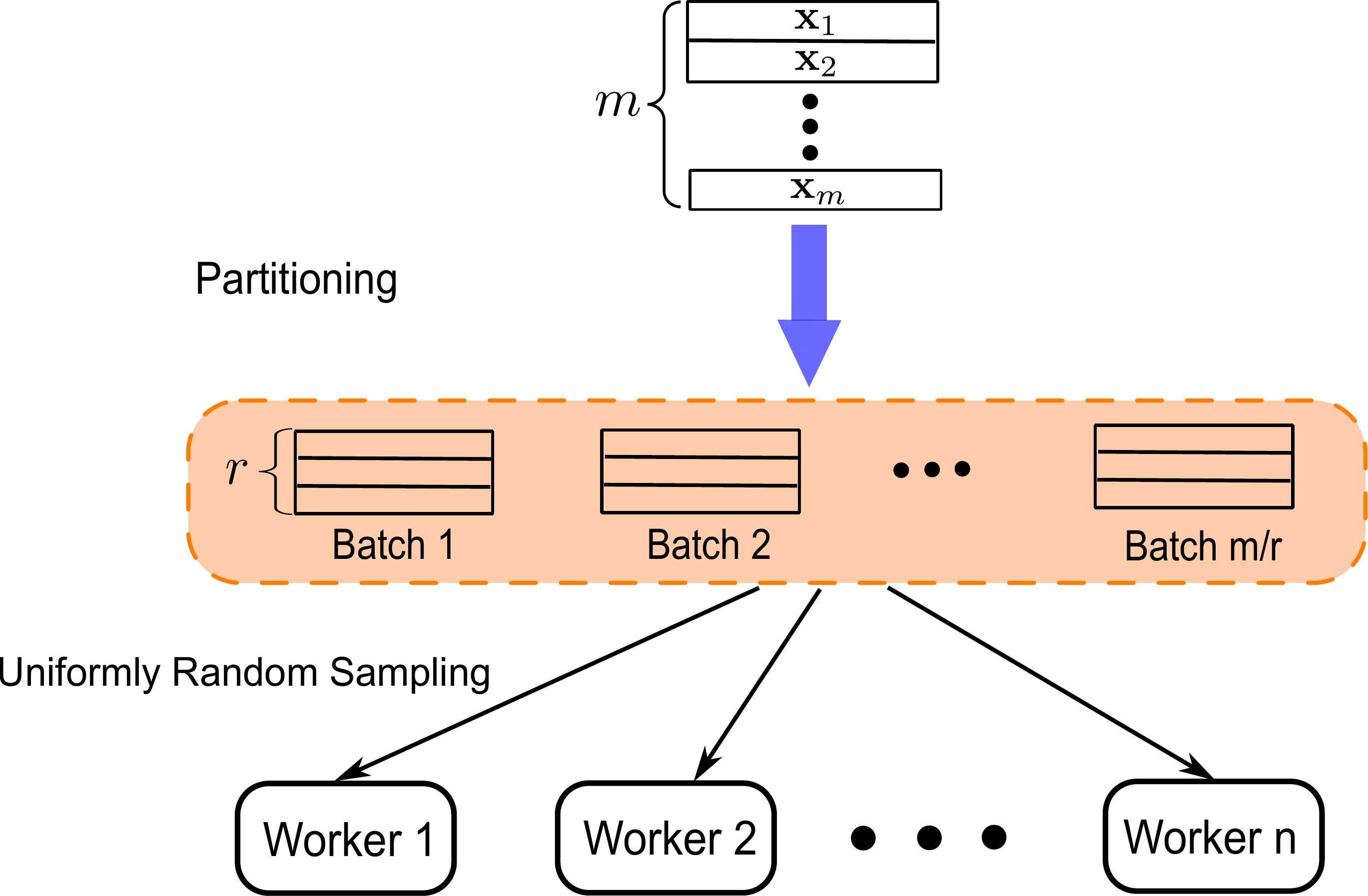}
  \caption{The data distribution of the proposed BCC scheme. The training dataset is evenly partition into $m/r$ batches of size $r$, from which each worker independently selects one uniformly at random.}
  \label{fig:BCC}
\end{figure}

\noindent {\bf Communication.} After computing the partial gradient $\vct{g}_j$ for all $j \in {\cal B}_{\sigma_i}$, Worker~$i$ computes a single message by summing them up i.e., 
\begin{align}
\vct{z}_i = \sum_{j \in {\cal B}_{\sigma_i}} \vct{g}_j,
\end{align}
and sends $\vct{z}_i$ to the master.

\noindent {\bf Data Aggregation at the Master.} When the master node receives the message from a worker, it discards the message if the master has received the result from processing the same batch before, and keeps the message otherwise. The master keeps collecting messages until the processing results from all data batches are received. Finally, the master reduces the kept messages to the final result by simply computing their summation.

We would like to note that the above BCC scheme is fully decentralized and coordination-free. Each worker selects its data batch independently of the other workers, and performs local computation and communication in a completely asynchronous manner. There is no need for any feedback from the master to the workers or between the workers. All these features make this scheme very simple to implement in practical scenarios.

\subsection{Near Optimal Performance Guarantees for BCC}
In this subsection, we theoretically analyze the BCC scheme, whose performance provides an upper bound on the minimum recovery threshold of the distributed GD problem, as well as an upper bound on the minimum communication load. To start, we state the main results of this paper in the following theorem, which characterizes the minimum recovery threshold and the minimum communication load to within a logarithmic factor.

\begin{theorem}\label{th:coupon}
For a distributed gradient descent problem of training $m$ data examples distributedly over $n$ worker nodes, we have
\begin{align}
\tfrac{m}{r} \leq K^*(r) \leq K_{\textup{BCC}}(r)= \lceil\tfrac{m}{r}\rceil H_{\lceil\tfrac{m}{r}\rceil}, \label{eq:RS} 
\end{align}
\begin{align}
\tfrac{m}{r} \leq L^*(r) \leq L_{\textup{BCC}}(r)= \lceil\tfrac{m}{r}\rceil H_{\lceil\tfrac{m}{r}\rceil}, \label{eq:CL}
\end{align}
for sufficiently large $n$, where $K^*(r)$ and $L^*(r)$ are the minimum recovery threshold and the minimum communication load respectively, $K_{\textup{BCC}}(r)$ and $L_{\textup{BCC}}(r)$ are the recovery threshold and the communication load achieved by the BCC scheme, and $H_t = \sum_{k=1}^t \frac{1}{k}$ is the $t$-th harmonic number. 
\end{theorem}


\begin{remark}
Given that $H_{\lceil\tfrac{m}{r}\rceil} \approx \lceil\tfrac{m}{r}\rceil \log (\lceil\tfrac{m}{r}\rceil)$, the results of Theorem~\ref{th:coupon} imply that for the homogeneous setting, the proposed BCC scheme simultaneously achieves a near minimal recovery threshold and communication load ( up to a logarithmic factor). $\hfill \square$
\end{remark}


\begin{remark}
As we mentioned before, other coding-based approaches~\cite{TLDK16,halbawi2017improving,raviv2017gradient} mostly focus on the worst-case scenario, resulting in a high recovery threshold e.g.~$K_{\textup{CR}}=m-r+1$.
\footnote{This is assuming $m=n$. We would like to point out that although designed for the worst-case, the fractional scheme proposed in \cite{TLDK16} may finish when the master collects results from less than $m-r+1$ workers. However, it only applies to the case where $r$ divides $m$.} In contrast, instead of focusing on worst-case scenarios, our proposed scheme aims at achieving the ``coverage'' of the partial computation results at the master, by collecting the computation of a much smaller number of workers (on average). As numerically demonstrated in Fig.~\ref{fig:compare} in Section~\ref{sec:intro}, The BCC scheme brings down the recovery threshold from $m-r+1$ to roughly $\frac{m}{r}\log \frac{m}{r}$.
$\hfill \square$
\end{remark}

\begin{remark}
In the coded computing schemes proposed in~\cite{TLDK16,halbawi2017improving,raviv2017gradient}, a linear combination of the locally computed partial gradients is carefully designed at each worker, such that the final gradient can be recovered at the master with minimum message sizes communicated by the workers. In the BCC scheme, each worker also communicates a message of minimum size, which is created by summing up the local partial gradients. As a result, BCC achieves a much smaller recovery threshold and hence can substantially reduces the total amount of network traffic. This is especially true when the dimension of the gradient is large, leading to significant speed-ups in the job execution. $\hfill \square$
\end{remark}

\begin{remark}
The coded schemes in~\cite{TLDK16,halbawi2017improving,raviv2017gradient} are designed to make the system robust to a fixed number of stragglers. Specifically, for a cluster with $s$ stragglers, a code can be designed such that the master can proceed after receiving $m-s$ messages, no matter which $s$ workers are slow. However, it is often difficult to predict the number of stragglers in a cluster, and it can change across iterations of the GD algorithm, which makes the optimal selection of this parameter for the coding schemes in~\cite{TLDK16,halbawi2017improving,raviv2017gradient} practically challenging. In contrast, our proposed BCC scheme is \emph{universal}, i.e., it does not require any prior knowledge about the stragglers in the cluster, and still promises a near-optimal straggler mitigation.    
$\hfill \square$
\end{remark}

\begin{proof}[Proof of Theorem~\ref{th:coupon}]
The lower bound $\frac{m}{r}$ in (\ref{eq:RS}) and (\ref{eq:CL}) is straightforward. They correspond to the best-case scenario where all workers the master hears from before completing the task, have mutually disjoint training examples. The upper bound in (\ref{eq:RS}) and (\ref{eq:CL}) is simultaneously achieved by the above described BCC scheme. To see this, we view the process of collecting messages at the master node as the classic coupon collector's problem (see e.g.,~\cite{ross2012first}), in which given a collection of $N$ types of coupons, we need to draw uniformly at random, one coupon at a time with replacement, until we collect all types of coupons. In this case, we have $\lceil \frac{m}{r}\rceil$ batches of training examples, from which each worker independently selects one uniformly at random to process. It is clear that the process of collecting messages at the master is equivalent to collecting coupons of $N = \lceil \frac{m}{r}\rceil$ types. As we know that the expected numbers of draws to collect all $N$ different types of coupons is $N H_N$, we use $N = \lceil \frac{m}{r}\rceil$ and reach the upper bound on the minimum recovery threshold. To characterize the communication load of the BCC scheme, we first note that since each worker communicates the summation of its computed partial gradients, the message size from each worker is the same as the size of the gradient computed from a single example. As a result, a communication load of $1$ is accumulated from each surviving worker, and the BCC scheme achieves a communication load that is the same as the achieved recovery threshold.
\end{proof}

Beyond the theoretical analysis, we also implement the proposed BCC scheme for distributed GD over Amazon EC2 clusters. In the next section, we describe the implementation details, and compare its empirical performance with two baseline schemes.

\subsection{Empirical Evaluations of BCC}
In this subsection, we present the results of experiments performed over Amazon EC2 clusters. In particular, we compare the performance of our proposed BCC scheme, with the following two schemes.
\begin{itemize}[leftmargin=*]
    \item uncoded scheme: In this case, there is no repetition in data among the workers and the master has to wait for all the workers to finish their computations. 
    \item cyclic repetition scheme of~\cite{TLDK16}: In this case, each worker processes $r$ training examples and in every iteration, the master waits for the fastest $m-r+1$ workers to finish their computations. 
\end{itemize}

\subsubsection{Experimental Setup}
We train a logistic regression model using Nesterov's accelerated gradient method. We compare the performance of the BCC, the uncoded and the cyclic repetition schemes on this task. We use Python as our programming language and MPI4py~\cite{dalcin2011parallel} for message passing across EC2 instances.
In our implementation, we load the assigned training examples onto the workers before the algorithms start. We measure the total running time via \texttt{Time.time()}, by subtracting the starting time of the iterations from the completion time at the master. In the $t$th iteration, the master communicates the latest model 
$\vct{w}_{t}$ to all the workers using \texttt{Isend()}, and each worker receives the updated model using \texttt{Irecv()}. 
In the cyclic repetition scheme, each worker sends the master a linear combination of the computed partial gradients, whose coefficients are specified by the coding scheme in~\cite{TLDK16}. In the BCC and uncoded schemes the workers simply send the summation of the partial gradients back to the master. When the master receives enough messages from the workers, it computes the overall gradient and updates the model. 


\noindent {\bf Data Generation.} We generate artificial data using a similar model to that of~\cite{TLDK16}. 
Specifically, we create a dataset consisting of $d$ input-output pairs of the form ${\bf D}=\{(\vct{x}_{1},y_{1}),(\vct{x}_{2},y_{2}),\ldots,(\vct{x}_d,y_d)\}$, where the input vector $\vct{x}_i \in \mathbb{R}^p$ contains $p$ features, and the output $y_i \in \{-1,1\}$ is the corresponding label. In our experiments we set $p = 8000$. To create the dataset, we first generate the true weight vector $\vct{w}^{*}$ whose coordinates are randomly chosen from $\{-1,1\}$. Then, we generate each input vector according to $\vct{x} \sim 0.5 \times \mathcal{N}(\vct{\mu}_{1},\vct{I})+0.5 \times \mathcal{N}(\vct{\mu}_{2},\vct{I})$ where $\vct{\mu}_{1}=\frac{1.5}{p}\vct{w}^{*}$ and $\vct{\mu}_{2}=\frac{-1.5}{p}\vct{w}^{*}$, and its corresponding output label according to $y \sim \mathit{Ber}(\kappa)$, with $\kappa=1/ (\exp(\vct{x}^{T}\vct {w}^{*})+1)$.


We run Nesterov's accelerated gradient descent distributedly for 100 iterations, using the aforementioned three schemes. We compare their performance in the following two scenarios:
\begin{itemize}[leftmargin=*]
\item scenario one: We use $51$ \textbf{t2.micro} instances, with one master and $n=50$ workers. We have $m=50$ data batches, each of which contains $100$ data points generated according to the aforementioned model. 
    
\item scenario two: We use $101$ \textbf{t2.micro} instances, with one master and $n=100$ workers. We have $m=100$ data batches, each of which contains $100$ data points.
\end{itemize}

\subsubsection{Results}
For the uncoded scheme, each worker processes $r = \frac{m}{n}$ data batches. For the cyclic repetition and the BCC schemes, we select the computational load $r$ based on the memory constraints of the instances so as to minimize the total running times.


\begin{figure}[htbp]
  \centering
  \includegraphics[width=0.48\textwidth]{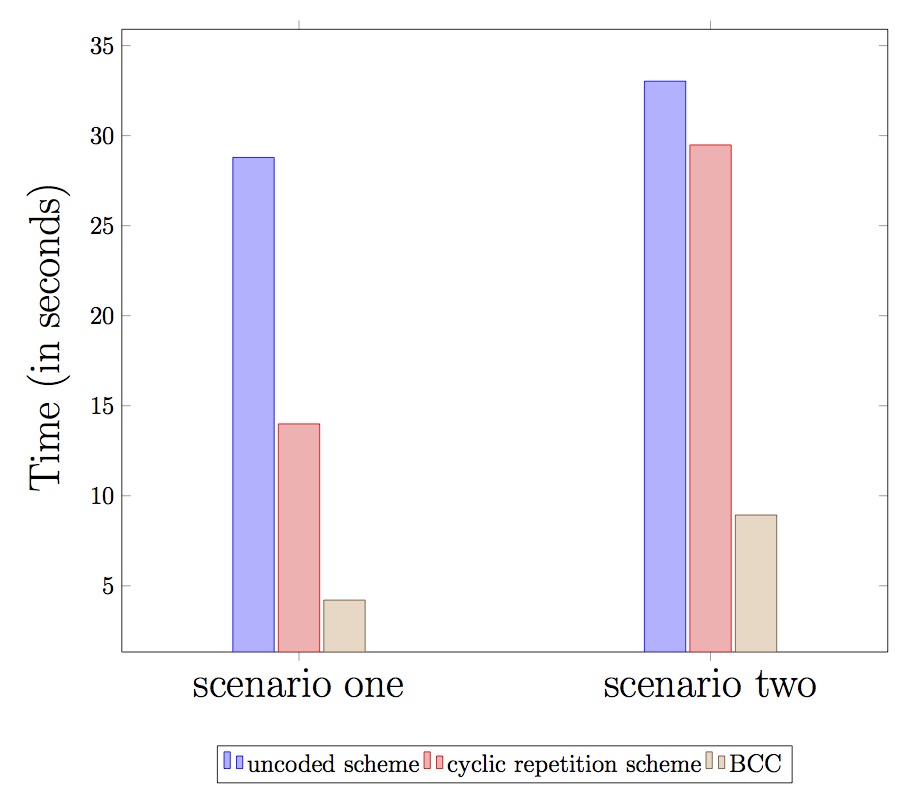}
  \caption{Running time comparison of the uncoded, the cyclic repetition, and the BCC schemes on Amazon EC2. In scenario one, we have $n=50$ workers, and $m=50$ data batches. In scenario two, we have $n=100$ workers, and $m=100$ data batches. Each data batch contains $100$ data points. In both scenarios, the cyclic repetition and the BCC schemes have a computational load of $r=10$.}
  \label{fig:run-time1}
  \vspace{-5mm}
\end{figure}



We plot the total running times 
of the three schemes in both scenarios in Fig.~\ref{fig:run-time1}. We also list the breakdowns of the running times for scenario one in Table~\ref{table:scenario one} and scenario two in Table~\ref{table:scenario two} respectively. Within each iteration, we measure the computation time as the maximum computation time among the workers whose results are received by the master before the iteration ends. After the last iteration, we add the computation times in all iterations to reach the total computation time. The communication time is computed as the difference between the total running time and the computation time.\footnote{Due to the asynchronous nature of the distributed GD, we cannot exactly characterize the time spent on computation and communication (e..g., often both are happening at the same time). The numbers listed in Tables I and II provide approximations of the time breakdowns.}

We draw the following conclusions from these results.
\begin{itemize}[leftmargin=*]
\item As we observe in Fig. \ref{fig:run-time1}, in scenario one, the BCC scheme speeds up the job execution by 85.4\% over the uncoded scheme, and 69.9\% over the cyclic repetition scheme. In scenario two, the BCC scheme speeds up the job execution by 73.0\% over the uncoded scheme, and 69.7\% over the cyclic repetition scheme. In scenario one, we observe the master waiting for on average $11$ workers to finish their computations, compared with $41$ workers for the cyclic repetition scheme and all $50$ workers for the uncoded scheme. In scenario two, we observe the master waiting for on average $25$ workers to finish their computations, compared with $91$ workers for the cyclic repetition scheme and all $100$ workers for the uncoded scheme.

    
\item As we note in Fig.~\ref{fig:run-time1}, the performance gains of both cyclic repetition and BCC schemes over the uncoded scheme become smaller with increasing number of workers. This is because that as the number of workers increases, in order to optimize the total running time, we need to also increase the computational load $r$ at each worker to maintain a low recovery threshold. However, due to the memory constraints at the worker instances, we cannot increase $r$ beyond the value $10$ to fully optimize the run-time performance.
    
\item  We observe from Table~\ref{table:scenario one} and Table~\ref{table:scenario two} that having a smaller recovery threshold benefits both the computation time and the communication time. While the BCC scheme and the cyclic repetition scheme have the same computational load at each worker, the computation time of BCC is much shorter since it needs to wait for a smaller number of workers to finish. On the other hand, lower recovery threshold of BCC yields a lower communication load that is directly proportional to the communication time. As a result, since in all experiments the communication time dominates the computation  time, the total running time of each scheme is approximately proportional to its recovery threshold.
\end{itemize}

\begin{table}
  \centering
  \scriptsize
  \begin{tabular}{|M{1cm}|M{1cm}|M{1.4cm}|M{1.1cm}|M{1cm}|}
    \hline
    scheme & recovery threshold & communication time (sec.) &computation time (sec.) &total running time (sec.) \\ \hline
    uncoded&50   &28.556&0.230&   28.786\\ \hline
    cyclic repetition&41&12.031  & 1.959   &13.990\\ \hline
    BCC&11&3.043 & 1.162&  4.205\\ \hline
  \end{tabular}
  \newline\newline
  \caption{Breakdowns of the running times  of the uncoded, the cyclic repetition, and the BCC schemes in scenario one.}
\label{table:scenario one}
\end{table}

\begin{table}
  \centering
  \scriptsize
  \begin{tabular}{|M{1cm}|M{1cm}|M{1.4cm}|M{1.1cm}|M{1cm}|}
    \hline
    scheme & recovery threshold          & communication time (sec.) &computation time (sec.) &total running time (sec.) \\ \hline
    uncoded&100   & 31.567   &1.453&   33.020\\ \hline
    cyclic repetition&91&24.698& 4.784   &29.482\\ \hline
    BCC&25&7.246 & 1.685&  8.931\\ \hline
  \end{tabular}
  \newline\newline
  \caption{Breakdowns of the running times  of the uncoded, the cyclic repetition, and the BCC schemes in scenario two.}
\label{table:scenario two}
\end{table}

\section{Extension to Heterogeneous Clusters}\label{sec:extension}
For distributed GD in heterogeneous clusters, workers have different computational and communication capabilities. In this case, the above proposed BCC scheme is in general sub-optimal due to its oblivion of network heterogeneity. In this section, we extend the above BCC scheme to tackle distributed DC over heterogeneous clusters. We also theoretically demonstrate that the extended BCC scheme provides an approximate characterization of the minimum job execution time.

\subsection{System Model}
In the heterogeneous setting, we consider an \emph{uncoded} communication scheme where after processing the local training examples, each worker communicates each of its locally computed partial gradients separately to the master. That is, Worker $i$, $i=1,\ldots,n$, communicates $\vct{z}_i = \{\vct{g}_j: j \in {\cal G}_i\}$ to the master. 
Under this communication scheme, the master computes the final gradient as soon as it collects the partial gradients computed from all examples. When this occurs, we say that \emph{coverage} is achieved at the master node.

We assume that the time required for Workers to process the local examples and deliver the partial gradients are independent from each other. We assume that this time interval, denoted by $T_i$ for Worker $i$, is a random variable with a shift-exponential distribution, i.e.,
\begin{align}
\text{Pr}[T_i \le t]=1-\exp\left(\tfrac{-\mu_i}{r_i}(t-a_i r_i)\right), \; t \geq a_i r_i.
\end{align}
Here, $\mu_i\geq 0$ and $a_i \geq 0$ are the fixed straggler and shift parameters of Worker~$i$. 

In this case, the total job execution time, or the time to achieve coverage at the master is given by
\begin{equation}
    T := \min \left\{t : \underset{i:T_i \leq t}{\cup}{\cal G}_i =\{1,\ldots,m\} \right\}.
\end{equation}


We are interested in characterizing the minimum average execution time in a heterogeneous cluster, which can be formulated as the following optimization problem.
\begin{align}
{\cal P}_1:\quad &\underset{\bf G}{\textup{minimize }}\mathbb{E}[T],\label{eq:cover}
\end{align}

In the rest of this section, we develop lower and upper bounds on the optimal value of ${\cal P}_1$.

\subsection{Lower and Upper Bounds on Optimal Value of ${\cal P}_1$}
To start, we first define the waiting time for the master to receive at least $s$ partial gradients (possibly with repetitions)
\begin{align}
\hat{T}(s) := \min \left\{t : \sum\limits_{i:T_i \leq t}r_i \geq s \right\}.
\end{align}
We also consider the following optimization problem 
\begin{align}
{\cal P}_2:\quad \underset{r_1,\ldots,r_n}{\textup{minimize }}\mathbb{E}[\hat{T}(s)].\label{eq:count}
\end{align}

For the master to collect all $m$ partial gradients, one computed from each training example, for any dataset placement, it has to receive at least $s \geq m$ partial gradients (possibly with repetitions) from the workers. Therefore, it is obvious that the coverage time $T$ cannot be shorter than $\hat{T}(m)$, and the optimal value $\underset{r_1,\ldots,r_n}{\min}\mathbb{E}[\hat{T}(m)]$ provides a lower bound on the optimal value of the coverage problem ${\cal P}_1$. For the above optimization problem ${\cal P}_2$, an algorithm is developed in~\cite{reisizadehmobarakeh2017coded} for distributed matrix multiplication on heterogeneous clusters. This algorithm obtains computation loads $r_1,\ldots,r_n$ that are asymptotically optimal in the large $n$ limit. Therefore, utilizing the results in~\cite{reisizadehmobarakeh2017coded}, we can obtain a good estimate of the optimal value $\underset{r_1,\ldots,r_n}{\min}\mathbb{E}[\hat{T}(s)]$.

It is intuitive that once we fix the work loads at the worker, i.e.,  $(r_1,r_2,\ldots,r_n)$, the time for the master to receive $s$ results $\hat{T}_s$ should increase as $s$ increases. We formally state this phenomenon in the following lemma.

\begin{lemma}[Monotonicity]\label{lemma:monotonic}
Consider an arbitrary dataset placement ${\bf G}$ where Worker~$i$ processes $|{\cal G}_i| = r_i$ training examples, for any $0 \leq s_1,s_2 \leq \sum_{i=1}^n r_i$, such that $s_1 \leq s_2$, we have 
\begin{align}
\mathbb{E}_{\bf G}[\hat{T}(s_1)] \leq \mathbb{E}_{\bf G}[\hat{T}(s_2)]. 
\end{align}
\end{lemma}

\begin{proof}
For a fixed dataset placement ${\bf G}$, we consider a particular realization of the computation times across the $n$ workers, denoted by ${\boldsymbol \delta} =(t_1,t_2,\ldots,t_n)$, where $t_i$ is the realization of $T_i$ for Worker~$i$ to process $r_i$ data points. We denote the realization of $\hat{T}(s)$ under ${\boldsymbol \delta}$ as $\hat{t}^{\boldsymbol \delta}(s)$. Obviously, for $s_1 \leq s_2$, we have $\hat{t}^{\boldsymbol \delta}(s_1) \leq \hat{t}^{\boldsymbol \delta}(s_2)$. Since this is true for all realizations ${\boldsymbol \delta}$, we have $\mathbb{E}_{\bf G}[\hat{T}(s_1)] \leq \mathbb{E}_{\bf G}[\hat{T}(s_2)]$. 
\end{proof}

To tackle the distributed GD problem over heterogeneous cluster, we generalize the above BCC scheme, and characterize the completion time of the generalized scheme using the optimal value of the above problem ${\cal P}_2$. The characterized completion time serves as an upper bound on the minimum average coverage time. Next, we state this result in the following theorem.
\begin{theorem}\label{theorem:bounds}
For a distributed gradient descent problem of training $m$ data examples distributedly over $n$ heterogeneous worker nodes, where the computation and communication time at Worker $i$ has an exponential tail with a straggler parameter $\mu_i$ and a shift parameter $a_i$, the minimum average time to achieve coverage is bounded as 
\begin{align}
    \min\limits_{\bf G}\mathbb{E}[T] &\geq \min\limits_{r_1,\ldots,r_n}\mathbb{E}[\hat{T}(m)], \\ 
    \min\limits_{\bf G}\mathbb{E}[T] &\leq \min\limits_{r_1,\ldots,r_n}\mathbb{E}[\hat{T}(\lfloor cm \log m \rfloor)]+1,
\end{align}
where $c = 2 + \frac{\log (a+H_n/\mu)}{\log m}$, $a =\max(a_1,\ldots,a_n)$, $\mu =\min(\mu_1,\ldots,\mu_n)$.
\end{theorem}

The proof of Theorem~\ref{theorem:bounds} is deferred to the appendix.

\begin{remark}
The above theorem, when combined with the results in~\cite{reisizadehmobarakeh2017coded} on evaluating $\underset{r_1,\ldots,r_n}{\min}\mathbb{E}[\hat{T}(s)]$, allows us to obtain a good estimate on the average minimum coverage time. Specifically, we can apply the results in \cite{reisizadehmobarakeh2017coded} to evaluate the lower and upper bounds in Theorem~\ref{theorem:bounds} for $s=m$ and $s=\lfloor cm \log m \rfloor$, respectively.
$\hfill \square$
\end{remark}

\begin{remark}
The upper bound on the average coverage time is achieved by a generalized BCC scheme, in which given the optimal data assignments $(r_1^*,\ldots,r_n^*)$ for ${\cal P}_2$ with $s \!=\! \lfloor cm \log m \rfloor$, Worker~$i$ independently selects $r_i^*$ examples uniformly at random. We emphasize that similar to the BCC data distribution policy in the homogeneous setting, the main advantages of the generalized BCC lies in its simplicity and decentralized nature. That is, each node selects the training examples randomly and independently from the other nodes, and we do not need to enforce a global plan for the data distribution. This also provides a scalable design so that when a new worker is added to the cluster, according to the updated dataset assignments computed from ${\cal P}_2$ with $n+1$ workers and $s = \lfloor cm \log m \rfloor$, each worker can individually adjust its workload by randomly adding or dropping some training examples, without needing to coordinate with the master or other workers. $\hfill \square$
\end{remark}

\vspace{-2mm}
\subsection{Numerical Results}
\vspace{-1mm}
We numerically evaluate the performance of the generalized BCC scheme in heterogeneous clusters, using the proposed random data assignment. In this case, we compute the optimal assignment $(r_1^*,\ldots,r_n^*)$ to minimize the average time for the master to collect $\lfloor m\log m \rfloor$ partial gradients. In comparison, we also consider a ``load balancing'' (LB) assignment strategy where the $m$ data points are distributed across the cluster based on workers' processing speeds, i.e., $r_i = \frac{\mu_i}{\sum \mu_i} m$.

We consider the computation task of processing $m=500$ examples over a heterogeneous cluster of $n=100$ workers. All workers have the same shift parameter $a_i = 20$, for all $i=1,\ldots,n$. The straggling parameter $\mu_i = 1$ for $95$ workers, and $\mu_i= 20$ for the remaining $5$ workers. As shown in Fig.~\ref{fig:numerical}, the computation of the LB assignment is long since the master needs to wait for every worker to finish. However, utilizing the proposed random assignment, the master can terminate the computation once it has achieved coverage, which significantly alleviates the straggler effect. As a result, the generalized BCC scheme reduces the average computation time by $29.28 \%$ compared with the LB scheme.

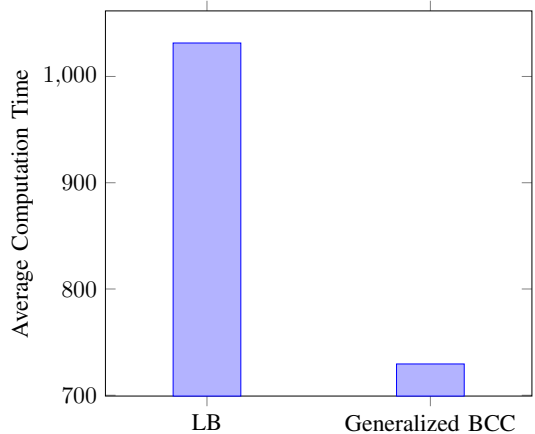
\begin{figure}[h]
\centering
\begin{tikzpicture}[scale=0.9] \begin{axis}[
ybar,
x=3.3cm,
bar width=1cm,
enlarge x limits={abs=1.5cm},
ylabel={Average Computation Time},
symbolic x coords={LB,Generalized BCC},
xtick=data,
]
\addplot coordinates {(LB,1031.5) (Generalized BCC,729.45)};
\end{axis}
\end{tikzpicture}
\caption{Illustration of the performance gain achieved by generalized BCC scheme for a heterogeneous cluster.}
\label{fig:numerical}
\end{figure}


\section{Conclusion}
We propose a distributed computing scheme, named batched coupon's collector (BCC), which effectively mitigates the straggler effect in distributed gradient descent algorithms. We theoretically illustrate that the BCC scheme is robust to the maximum number of stragglers to within a logarithmic factor. We also empirically demonstrate the performance gain of BCC over baseline straggler mitigation strategies on EC2 clusters. Finally, we generalize the BCC scheme to minimize the job execution time over heterogeneous clusters.

\bibliographystyle{IEEEtran}
\bibliography{ref-abb}


\section*{Appendix \\ Proof of Theorem~\ref{theorem:bounds}}

Before starting the formal proof of Theorem~{\ref{theorem:bounds}}, we first state a result for the coupon collector's problem that will become useful later. We denote the random variable that represents the minimum number of coupons one needs to collect before obtaining all $m$ types of coupons as $M (M \geq m)$, and present an upper bound on the tail probability in the following lemma.

\begin{lemma}[Theorem 1.23 in \cite{auger2011theory}]\label{lemma:tail}
$\textup{Pr}(M \geq (1+\epsilon) m \log m) \leq \frac{1}{m^{\epsilon}}$, for any $\epsilon \geq 0$.
\end{lemma}

We prove Theorem~\ref{theorem:bounds} in two steps. In the first step, we propose a generalized BCC scheme, for which no batching operation is performed on the dataset, and the workers simply sample the examples to process uniformly at random. In the second step, we analyze the average execution time of the generalized BCC scheme. To start, we obtain an estimate of the number of partial gradients the master receives before coverage is achieved (analogous to the recovery threshold in the homogeneous setting). Then, conditioned on the value of this number, we derive an upper bound on the average coverage time, which is obviously also an upper bound on the minimum coverage time over all schemes. 

\begin{proof}[Proof of Theorem~\ref{theorem:bounds}]
For any dataset placement ${\bf G}$, whenever the collected partial gradients at the master cover the results from all $m$ data points, the master must have already collected at least $m$ partial gradients. Therefore, $\min\limits_{r_1,\ldots,r_n}\mathbb{E}[\hat{T}(m)] \leq \min\limits_{\bf G}\mathbb{E}[T]$.

Consider the optimization problem
\begin{align*}
\underset{r_1,\ldots,r_n}{\textup{minimize }}\mathbb{E}[\hat{T}(\lfloor cm \log m \rfloor)]    
\end{align*}
where $c$ is specified in the statement of Theorem~\ref{theorem:bounds}. Assume the optimal task assignment is given by
\begin{align}
(r_1^*,\ldots,r_n^*) &= \arg\min\limits_{r_1,\ldots,r_n}\mathbb{E}[\hat{T}(\lfloor cm \log m \rfloor)].
\end{align}

Now given $(r_1^*,\ldots,r_n^*)$, we consider a specific data distribution ${\bf G}_0$ in which Worker~$i$, selects $r_i^*$ out of $m$ training examples without replacement (independently and uniformly at random) and processes them locally. Next, we show that using this particular placement ${\bf G}_0$, we can achieve an average coverage time $\mathbb{E}[T]$ that is at most $\min\limits_{r_1,\ldots,r_n}\mathbb{E}[\hat{T}(\lfloor cm \log m \rfloor)]+1$. 

First, we consider a relaxed data distribution strategy ${\bf G}_1$ in which Worker~$i$, independently, and uniformly at random selects $r_i^*$ data points with replacement, and processes them locally. That is, ${\bf G}_1$ allows each worker to process an example more than once. It is obvious that
\begin{align}\label{eq:replacement}
\mathbb{E}_{{\bf G}_0}[T]\leq \mathbb{E}_{{\bf G}_1}[T].
\end{align}

We note that when using the data distribution  ${\bf G}_1$ the process of receiving partial gradients at the master mimics the process of collecting coupons in the coupon collector's problem. We define a random variable $W (W \geq m)$ as the minimum number of partial gradients (possibly with repetition) the master receives before it reaches coverage. We note that $W$ is statistically equivalent to the minimum number of coupons one needs to collect in the coupon collector's problem. In what follows, we only consider the case where the coverage can be achieving using the messages sent by all $n$ nodes (or the computation can be successfully executed), i.e., $W \leq \sum_{i=1}^n r_n^*$.

Taking expectation conditioned on the value of $W$, we have 
\begin{align}
&\hspace{1mm}\mathbb{E}_{{\bf G}_1}[T] \nonumber\\
&= \textup{Pr}(m \!\leq\! W \leq cm \log m)\mathbb{E}_{{\bf G}_1}[T|m \leq \!W\! \leq cm \log m] \nonumber \\
&\hspace{1mm}+\!\!\textup{Pr}(cm \log m \!<\! W \!\leq\!\! \sum_{i=1}^{n}\! r_i^{*}\!)\mathbb{E}_{{\bf G}_1}\!\!\left[\!T|cm \log m \!<\! W \!\!\leq\!\! \sum_{i=1}^{n} \!r_i^{*}\!\right]\\
&\leq \mathbb{E}_{{\bf G}_1}[\hat{T}(W)|m \leq W \leq \lfloor cm \log m \rfloor] \nonumber \\
&\hspace{1mm}+\! \textup{Pr}(W \!>\! cm \log m)\mathbb{E}_{{\bf G}_1}\!\!\left[\!\hat{T}(W)|cm \log m \!<\! W \!\leq\! \sum_{i=1}^{n} r_i^{*}\!\!\right]\\
&\overset{(a)}{\leq} \mathbb{E}_{{\bf G}_1}[\hat{T}(W)|m \leq W \leq \lfloor cm \log m \rfloor] \nonumber\\
&\hspace{1mm}+ \frac{1}{m^{c-1}}\mathbb{E}_{{\bf G}_1}\left[\hat{T}(W)|cm \log m < W \leq \sum_{i=1}^{n} r_i^{*}\right]\\
&\overset{(b)}{\leq} \mathbb{E}_{{\bf G}_1}\![\hat{T}(\lfloor cm \log m \rfloor)] \!+\! \frac{1}{m^{c-1}}\mathbb{E}_{{\bf G}_1}\!\left[\hat{T}\left(\sum_{i=1}^{n} r_i^{*}\right)\right]\\
&= \mathbb{E}_{{\bf G}_1}[\hat{T}( \lfloor cm \log m \rfloor)]\nonumber \\
&\hspace{1mm}+\frac{1}{m^{c-1}}\mathbb{E}_{{\bf G}_1}[\max(T_1,T_2,\ldots,T_n)]\\
&\overset{(c)}{\leq} \mathbb{E}_{{\bf G}_1}[\hat{T}( \lfloor cm \log m \rfloor)] \nonumber \\
&\hspace{1mm}+ \frac{1}{m^{c-1}}\mathbb{E}_{{\bf G}_1}[\max(\bar{T}_1,\bar{T}_2,\ldots,\bar{T}_n)] \\
&= \mathbb{E}_{{\bf G}_1}[\hat{T}( \lfloor cm \log m \rfloor)] +\frac{r^*}{m^{c-1}} \left(a + \frac{H_n}{\mu}\right)\\
&\leq  \mathbb{E}_{{\bf G}_1}[\hat{T}( \lfloor cm \log m \rfloor)] +\frac{a+\frac{H_n}{\mu}}{m^{c-2}}\\
&\overset{(d)}{=} \min\limits_{r_1,\ldots,r_n}\mathbb{E}[\hat{T}(\lfloor cm \log m \rfloor)]+1,\label{eq:countBound}
\end{align}
where step (a) is due to Lemma~\ref{lemma:tail}, and step (b) results from Lemma~\ref{lemma:monotonic} in Section~\ref{sec:extension}. In step (c), $\bar{T}_1,\bar{T}_2,\ldots,\bar{T}_n$ are i.i.d. random variables with the shift-exponential distribution
\begin{align}
\text{Pr}[\bar{T}_i \le t]=1-\exp\left(\tfrac{-\mu}{r^{*}}(t-a r^{*})\right), \; t \geq a r^{*},
\end{align}
for all $i=1,\ldots,n$, where $\mu = \min(\mu_1,\ldots,\mu_n)$, $a = \max(a_1,\ldots,a_n)$, and $r^* = \max(r^*_1,\ldots,r^*_n)$. Step (d) is because that we choose $c = 2 + \frac{\log (a+H_n/\mu)}{\log m}$.

Finally, we have from (\ref{eq:replacement}) and (\ref{eq:countBound}) that $\min\limits_{\bf G}\mathbb{E}[T] \leq \mathbb{E}_{{\bf G}_0}[T]\leq \mathbb{E}_{{\bf G}_1}[T] \leq  \min\limits_{r_1,\ldots,r_n}\mathbb{E}[\hat{T}(\lfloor cm \log m \rfloor)]+1$. 
\end{proof}

\end{document}